\documentclass{article}

\usepackage{arxiv}

\usepackage[utf8]{inputenc} % allow utf-8 input
\usepackage[T1]{fontenc}    % use 8-bit T1 fonts
\usepackage{hyperref}       % hyperlinks
\usepackage{url}            % simple URL typesetting
\usepackage{booktabs}       % professional-quality tables
\usepackage{amsfonts}       % blackboard math symbols
\usepackage{nicefrac}       % compact symbols for 1/2, etc.
\usepackage{microtype}      % microtypography
\usepackage{lipsum}

\usepackage{times}
\usepackage{fancyhdr,graphicx,amsmath,amssymb,amsthm}
\usepackage[ruled,lined]{algorithm2e}

\usepackage{multirow}
\usepackage[table,xcdraw]{xcolor}
\usepackage{tablefootnote}
\usepackage{float}

\include{pythonlisting}
\newtheorem{theorem}{Theorem}

\title{Large-scale Causal Approaches to Debiasing Post-click Conversion Rate Estimation with Multi-task Learning}

\author{
  Wenhao Zhang\thanks{Both authors contributed equally to this research. Authorship order determined by coin flip.} $^{\,}$      
  \thanks{Work done while interning at Alibaba Group.}\\
  University of California, Los Angeles\\
  \texttt{wenhaoz@ucla.edu} \\
   \And
  Wentian Bao$^{*}$ \\
  Alibaba Group \\
  \texttt{wentian.bwt@alibaba-inc.com} \\
       \And
  Xiao-Yang Liu \\
  Columbia University \\
  \texttt{xl2427@columbia.edu} \\
     \And
  Keping Yang \\
  Alibaba Group \\
  \texttt{shaoyao@alibaba-inc.com} \\
     \And
  Quan Lin \\
  Alibaba Group \\
  \texttt{tieyi.lq@alibaba-inc.com} \\
     \And
  Hong Wen \\
  Alibaba Group \\
  \texttt{qinggan.wh@alibaba-inc.com} \\
      \And
  Ramin Ramezani \\
  University of California, Los Angeles \\
  \texttt{raminr@ucla.edu} \\
}

\begin{document}
\maketitle

\begin{abstract}
Post-click conversion rate (CVR) estimation is a critical task in e-commerce recommender systems. This task is deemed quite challenging under industrial setting with two major issues: 1) selection bias caused by user self-selection, and 2) data sparsity due to the rare click events. A successful conversion typically has the following sequential events: "exposure $\rightarrow$ click $\rightarrow$ conversion". Conventional CVR estimators are trained in the click space, but inference is done in the entire exposure space. 
They fail to account for the causes of the missing data and treat them as missing at random. 
Hence, their estimations are highly likely to deviate from the real values by large.  In addition, the data sparsity issue can also handicap many industrial CVR estimators which usually have large parameter spaces.

In this paper, we propose two principled, efficient and highly effective CVR estimators for industrial CVR estimation, namely, \textbf{Multi-IPW} and \textbf{Multi-DR}. The proposed models approach the CVR estimation from a causal perspective and account for the causes of missing not at random. In addition, our methods are based on the multi-task learning framework and mitigate the data sparsity issue. Extensive experiments on industrial-level datasets show that our methods outperform the state-of-the-art CVR models.
\end{abstract}

% keywords can be removed

\section{Introduction}

Selection bias is a widely-recognized issue in recommender systems \cite{schnabel_recommendations_2016,ma_entire_2018,de2014reducing}. For example, music stream services usually suggest genres that have positive user feedbacks (e.g., favorite, share, and buy, etc.), and selectively ignore the ones that are rarely exposed to users \cite{van2013deep}. In this paper, we study the selection bias that exists in the post-click conversion rate (CVR) estimation. 

\begin{figure}[ht]
    \centering
    \includegraphics[width=90mm]{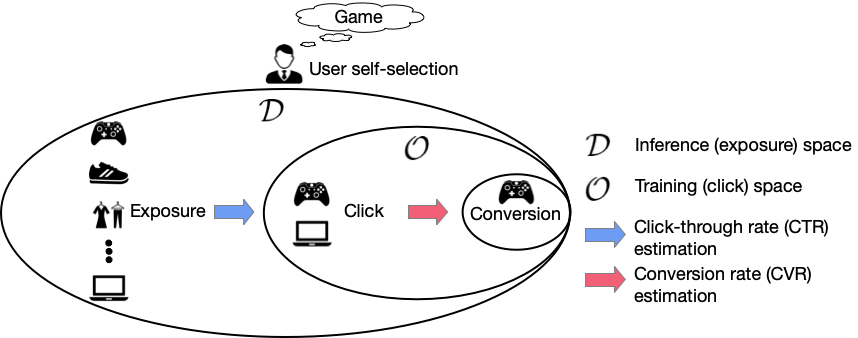}
    \caption{Illustration of the selection bias issue in conventional conversion rate (CVR) estimation. The training space of conventional CVR models is the click space $\mathcal{O}$, whereas the inference space is the entire exposure space $\mathcal{D}$.  The discrepancy of data distribution between $\mathcal{O}$ and $\mathcal{D}$ leads to selection bias in conventional CVR models.}
    \label{fig:selection_bias_illustration}
\end{figure}

\subsubsection*{Problem formulation}

Post-click conversion rate (CVR) estimation is a critical task in e-commerce recommender systems \cite{zhu2017optimized,yamaguchi2016web}. A typical e-commerce transaction has the following sequential events: "exposure $\rightarrow$ click $\rightarrow$ conversion" \cite{ma_entire_2018}. Post-click conversion rate indicates the probability of transitions from click to conversion. Typically,  when training CVR models, we only include the items that customers clicked on as we are unaware of the conversion feedback  of the items that are not clicked by customers \cite{huang2007correcting}. 
Bear in mind, not clicking on an item does not necessarily indicate the customer is not interested in purchasing it. Customers may unconsciously skip certain items that might be interesting to them. Fig. \ref{fig:selection_bias_illustration} reveals that the exposure space $\mathcal{D}$  is a super set of the click space  $\mathcal{O}$. Selection bias occurs when conventional CVR models are trained in the click space, and the predictions are made in the entire exposure space (see Fig. \ref{fig:selection_bias_illustration}) \cite{ma_entire_2018}. Intuitively, data in the click space is drawn from the entire exposure space and is biased by the user self-selection. Therefore, the data distribution in the click space is systematically different from the one in the exposure space. This inherent discrepancy leads to data that is missing not at random (MNAR), and selection bias in the conventional CVR models \cite{de2014reducing,little2019statistical,little2002statistical, bareinboim2014recovering}.

In summary, we identify two practical issues that make CVR estimation quite challenging in industrial-level recommender systems:
\begin{itemize}
    \item \textbf{Selection bias}: The systematic difference between data distributions in training space $\mathcal{O}$  and inference space $\mathcal{D}$ biases conventional CVR models \cite{steck2010training,huang2007correcting,ai2018unbiased, bareinboim2014recovering}. This bias usually causes the performance degradation. 
    \item \textbf{Data sparsity}: This issue occurs since clicks are relatively rare events (we have a CTR of 5.2\% in the production dataset and 4\% in the public dataset). Conventional CVR models are typically trained only using data in the click space. Therefore, the number of training samples may not be sufficient for the large parameter space. In the experiments, we have 0.6 billion data samples vs. 5.3 billion parameters for production dataset, and 4.3 million samples vs. 2.6 billion parameters for public dataset (see Section \ref{section:dataset}) \cite{lee2012estimating,wang2018billion}. 
\end{itemize}

To simplify the debiasing task of CVR estimation, we assume the exposure space is the entire item space we are interested in (see Fig. \ref{fig:selection_bias_illustration}) \cite{ma_entire_2018}. Such a relaxation is also made based on the postulation that most items are exposed at least once.  Table \ref{table:data_stats} shows that our dataset contains 81.5 million items and 11.5 billion exposures,  i.e., each item is exposed, on  average, about 150 times.

To address the critical issues of selection bias and data sparsity in the CVR estimation, we take a causal perspective and develop causal methods in a multi-task learning framework. We propose two principled, efficient and highly effective CVR estimators, namely, \textbf{Multi}-task \textbf{I}nverse \textbf{P}ropensity \textbf{W}eighting estimator (\textbf{Multi-IPW}) and \textbf{Multi}-task \textbf{D}oubly \textbf{R}obust estimator (\textbf{Multi-DR}). Our methods are designed for unbiased CVR estimation. They also account for the data sparsity issue.

The main contributions of this paper are summarized as follows:
\begin{itemize}
    \item To the best of our knowledge, this is the first paper that combines IPW-based and DR-based methods with multi-task learning. From a causal perspective, we aim to tackle the well-recognized issues (i.e., selection bias and data sparsity) in CVR estimation in concert. 
    \item We highlight that the state-of-the-art CVR model, ESMM \cite{ma_entire_2018}, is biased. Different from existing works, our methods adjust for MNAR data, and deal with the selection bias in a principled way. Meanwhile, we give mathematical proofs that the proposed methods are theoretically unbiased. The empirical study shows our approaches outperform ESMM and several state-of-the-art causal models, and demonstrates the efficiency of our methods in real industrial setting.
\end{itemize}

\section{Related works}
In this section, we review several existing works that attempt to tackle the selection bias issue in recommender systems. Meanwhile, we summarize how our methods are different from  prior works.

Ma \textit{et al.} \cite{ma_entire_2018} proposed the Entire Space Multi-task Model (ESMM) to remedy selection bias and data sparsity issues in the conversion rate (CVR) estimation. ESMM is trained in the entire exposure space, and it formulates CVR task as two auxiliary tasks, i.e., click-through rate (CTR) and click-through \& conversion rate (CTCVR) estimations. However, we argue that ESMM is biased. The details of our argument are presented in Section \ref{section:esmm_biased}.

Causal inference offers a way to adapt for the data generation process when we attempt to restore the information from MNAR data \cite{little2002statistical}. Schnabel \textit{et al}. \cite{schnabel_recommendations_2016} proposed an IPW-based estimator for training and evaluating recommender systems from biased data. 
IPW-based models may still be biased if the propensities are not accurately estimated. Wang \textit{et al.} \cite{wang_doubly_2019} proposed a doubly robust (DR) joint learning approach for estimating item ratings that are MNAR. Doubly robust estimator combines the IPW-based methods with an imputation model that estimates the prediction error for the missing data. When the propensities are not accurately learned, DR estimator can still enjoy unbiasedness as long as its imputation model is accurate. However, the existing DR-based methods are not devised for CVR estimation, hence fail to account for the severe data sparsity issue that widely exists in the CVR estimation. In addition, such a joint learning approach is not efficient in industrial setting (see Fig. \ref{fig:radar_chart}).

To summarize, our approach differs from aforementioned methods in three aspects: 1) The problems are different. We developed our methods for CVR estimation in e-commerce system, while they focus on the rating prediction \cite{lu2018coevolutionary}. 2) The challenges are different. we design our models to address the selection bias and data sparsity issues, while they only consider the former (ESMM considers both). 3) The methods are different. we integrate multi-task framework with causal approaches. Specifically, We co-train propensity model, imputation model and prediction model simultaneously with deep neural networks, while they train these modules separately or alternatively, and usually with models such as linear regression or matrix factorization \cite{gopalan2014content,guo2017deepfm,rendle2010factorization,he2016fast}. We will further justify our design in Section 3 and report the performance improvement in Section 5.

 \section{Causal CVR Estimators with multi-task learning}
\label{section:industrial}

\subsection{Preliminary}
Let $\mathcal{U}=(u_1, u_2, ..., u_N)$ be a set of $N$ users and $\mathcal{I}=(i_1, i_2,...,i_M)$ be a set of $M$ items, $\mathcal{D}=\mathcal{U} \times \mathcal{I}$ be the user-item pairs, $\mathbf{R} \in \mathbb{R}^{N \times M}$ be the true conversion label matrix where each entry $r_{u,i} \in \{0,1\}$, 
and $\mathbf{\hat{R}} \in \mathbb{R}^{N \times M}$ be the predicted conversion score matrix where each entry $\hat{r}_{u,i} \in [0,1]$. Then, the \textit{Prediction inaccuracy} $\mathcal{P}$ over all user-item pairs can be formulated as follows,
\begin{equation}
    \mathcal{P} = \mathcal{P}(\mathbf{R}, \mathbf{\hat{R}})= \frac{1}{|\mathcal{D}|}\sum_{(u,i) \in \mathcal{D}}e(r_{u,i}, \hat{r}_{u,i}),
\label{eq:P}
\end{equation}
where $e(r_{u,i}, \hat{r}_{u,i})=-r_{u,i}\log(\hat{r}_{u,i})-(1-r_{u,i})\log(1-\hat{r}_{u,i})$. 
 
Let $\mathbf{O} \in \{0,1\}^{\mathcal{U} \times \mathcal{I}}$ 
be the \textit{indicator matrix} where each entry $o_{u,i}$ is an observation indicator: $o_{u,i} =1$ if a user $u$ clicks on item $i$, $o_{u,i} =0$ otherwise. Since the clicks are subjective to certain unobserved factors (e.g., users latent interests), such user self-selection process generates MNAR data \cite{rubin1976inference,enders2010applied}. Naive CVR estimators are trained only in the click space $\mathcal{O}=\{(u,i)|o_{u,i} =1, (u,i) \in \mathcal{D}$\}. Let $\mathbf{R}^{\text{obs}}$ and $\mathbf{R}^{\text{mis}}$ be the set of conversion labels that are present and absent in $\mathcal{D}$. We evaluate these naive CVR models by averaging the cross-entropy loss over the observed data \cite{schnabel_recommendations_2016,wang_doubly_2019},
\begin{equation}
\begin{aligned}
    \mathcal{E}^{\text{Naive}} &= \mathcal{E}(\mathbf{R}^{\text{obs}}, \mathbf{\hat{R}}) \\
    &= \frac{1}{|\mathcal{O}|}\sum_{(u,i) \in \mathcal{O}}e(r_{u,i}, \hat{r}_{u,i})\\
    &= \frac{1}{|\mathcal{O}|}\sum_{(u,i) \in \mathcal{D}}o_{u,i}e(r_{u,i}, \hat{r}_{u,i}),
\end{aligned}    
\end{equation}
where $|\mathcal{O}| = \sum_{(u,i) \in \mathcal{D}}o_{u,i}$.

%%Compressed writing%%
We say a CVR estimator $\mathcal{M}$ is \textit{unbiased} when the expectation of the estimated prediction inaccuracy over $\mathbf{O}$ equals to the prediction inaccuracy $\mathcal{P}$, i.e.,  $\text{Bias}^\mathcal{M}=|\mathbb{E}_{\mathbf{O}}[\mathcal{E}^\mathcal{M}] - \mathcal{P}| = 0$, otherwise it is \textit{biased}.
If data is MNAR, $|\mathbb{E}_{\mathbf{O}}[\mathcal{E}_\mathcal{M}] - \mathcal{P}| \gg 0$ \cite{mohan2013graphical}.

\subsection{Is ESMM an Unbiased CVR Estimator?}
\label{section:esmm_biased}

\begin{figure}[t]
    \centering
    \includegraphics[width=65mm]{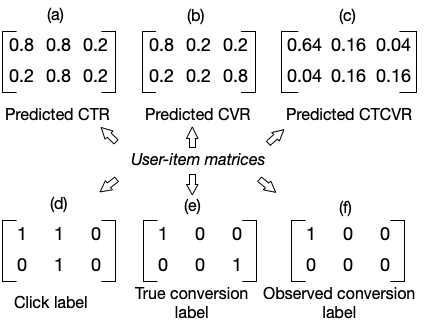}
    \caption{A toy example that demonstrates ESMM is biased.}
    \label{fig:esmm_bias}
\end{figure}

In this section, we demonstrate that ESMM, the state-of-the-art CVR estimator in practice, is essentially biased, though the author claim in the paper that the model eliminates the selection bias \cite{ma_entire_2018}. We formulate the estimation bias of ESMM, and prove it is not theoretically unbiased by giving a counter example. 

Let $e^{\text{CTR}}_{u,i},e^{\text{CVR}}_{u,i},e^{\text{CTCVR}}_{u,i}, (u,i) \in \mathcal{D}$, be the cross-entropy losses of CTR, CVR, and CTCVR tasks. Then we have,
\begin{equation}
\begin{aligned}
    &\text{Bias}^{\text{ESMM}}
    =|\mathbb{E}_{\mathbf{O}}[\mathcal{E}^{\text{ESMM}}] - \mathcal{P}|\\
    &= \biggr|\mathbb{E}_{\mathbf{O}}\biggr[\frac{1}{|\mathcal{D}|}\sum_{(u,i) \in \mathcal{D}}(e^{\text{CTR}}_{u,i} + e^{\text{CTCVR}}_{u,i})\biggr] - \frac{1}{|\mathcal{D}|}\sum_{(u,i) \in \mathcal{D}}e_{u,i}^{\text{CVR}} \biggr|\\
    &= \biggr|\frac{1}{|\mathcal{D}|}\sum_{(u,i) \in \mathcal{D}}(e^{\text{CTR}}_{u,i} + e^{\text{CTCVR}}_{u,i}) - \frac{1}{|\mathcal{D}|}\sum_{(u,i) \in \mathcal{D}}e_{u,i}^{\text{CVR}}\biggr|\\
    &= \frac{1}{|\mathcal{D}|}\biggr|\sum_{(u,i) \in \mathcal{D}}(e^{\text{CTR}}_{u,i} + e^{\text{CTCVR}}_{u,i} - e_{u,i}^{\text{CVR}})\biggr|.
\end{aligned}
\label{eq:esmm_biased}
\end{equation}
We can easily verify that $\text{Bias}^{\text{ESMM}} > 0$ using the counter example in Fig. \ref{fig:esmm_bias}. Note that to be theoretically unbiased, ESMM should satisfy $|\mathbb{E}_{\mathcal{D}}[\mathcal{E}^{\text{ESMM}}] - \mathcal{P}|=0,\forall \mathcal{D}$.  Therefore, we conclude that ESMM cannot ensure unbiased CVR estimation.

\subsection{A causal perspective to unbiased CVR estimation}

Recall that selection bias in CVR estimation comes from the fact that models are trained over the click space $\mathcal{O}$, whereas the predictions are made over the exposure space $\mathcal{D}$ (See Fig. \ref{fig:selection_bias_illustration}). 
Ideally, we can remove the selection bias by building our CVR estimators using a dataset where the conversion labels of all the items are known. In the language of causal inference, it is equivalent to training CVR estimators on a "do dataset", where causal intervention is applied on click event during the data generation process. Specifically, users are "forced" to click on every item in the exposure space $\mathcal{D}$ and further make their purchase decisions. 
Note that the training space is the same as the inference space in the "do dataset". Hence, the selection bias is eliminated. Intuitively, we can also understand how causal intervention removes the bias in  Fig. \ref{fig:causal_graph}. $Z$ denotes the self-selection factors that affect both click events and conversion events.  For example, $Z$ can be the purchase interest  or price discount that customers consider in online-shopping. In causal inference, we refer $Z$ as "confounder(s)" that biases the CVR  inference \cite{wang2018deconfounded}. Once the causal intervention is applied on the click event (i.e., users are  forced to click on all exposed items),  $Z$ has no control over user click behaviors. It means that we have successfully removed the confounder $Z$ which biases our CVR estimation \cite{pearl2009causal,pearl2018book,wang2018deconfounded,pearl2000causality,morgan2015counterfactuals}.

\begin{figure}[t]
    \centering
    \includegraphics[width=80mm]{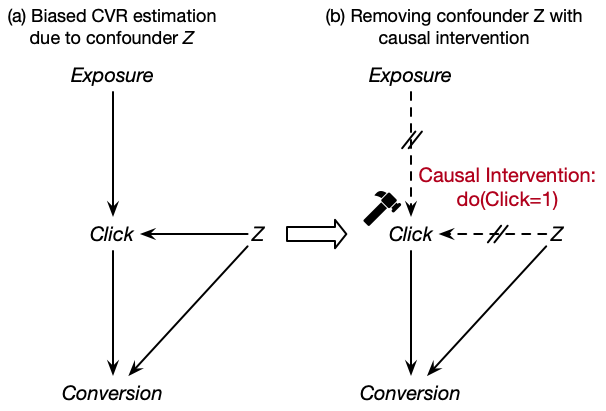}
    \caption{This causal graph formulate CVR estimation as a causal problem. \cite{pearl2018book} In (a), $Z$ is a confounder that affects both clicks and purchases, and it biases the inference. In (b), we apply intervention on click events (do(Click)=1). Once users are "forced" to click on each exposed item, $Z$ has no control over user click behaviors. Note the absence of the arrow from $Z$ to \textit{Click}. Hence, we have successfully removed the confounder $Z$, and the selection bias \cite{pearl2009causal,pearl2018book,wang2018deconfounded,pearl2000causality,morgan2015counterfactuals}.}
    \label{fig:causal_graph}
\end{figure}

Apparently, this "do dataset" generated in this imaginary intervention experiment cannot be obtained in reality. Now the challenge is how to train our CVR estimators on the observed dataset $\mathcal{O}$ as if we do on the "do dataset". In the following sections, we will discuss two estimators that can achieve unbiased CVR prediction with the data that are MNAR.

As for ESMM and its variants, they are directly trained over the entire exposure space, which is different from the "do dataset". In "do dataset", we apply causal intervention on the click events during the data generation process to obtain the true conversion label matrix (see matrix \textit{e} in Fig. \ref{fig:esmm_bias}). Without such intervention, the conversion labels in the dataset are still missing not at random. For example, in Fig. \ref{fig:esmm_bias}, matrix \textit{f} is obtained without causal intervention. The conversion label of the last entry is observed as 0 because its click label is 0, whereas the true conversion label is 1 in the matrix \textit{e}. We can also understand the weakness of ESMM and its variant from using causal graph (see Fig. \ref{fig:causal_graph}), CVR estimators cannot achieve unbiased CVR estimation as long as they fail to remove the confounder $Z$. 

\subsection{Multi-task Learning Module}
\label{subsection:mtl}

To address the data sparsity issue, we adopt the philosophy of multi-task learning and introduce an auxiliary CTR task \cite{wang2019multi}. The multi-task learning module exploits the typical sequential events in e-commerce recommender system, i.e., "exposure $\rightarrow$ click $\rightarrow$ conversion", and chains the main CVR task with the auxiliary CTR task. The amount of training data in CTR task is generally larger than that in CVR task by $1 \sim 2$ order of magnitudes (see Table \ref{table:data_stats}), thus CTR task trains the large volume of model parameters more sufficiently. Besides, the feature representation learned in the CTR task is shared with the CVR task, which makes the CVR model benefit from the extra information via parameter sharing. Hence, the data sparsity issue is remedied \cite{hadash2018rank, ma_entire_2018, ni2018perceive}. 

Meanwhile, multi-task learning is also perceived as being cost-effective in training phase \cite{ma_entire_2018}. Specifically, multi-task learning co-trains multiple tasks simultaneously as if they were one task. This mechanism can potentially reduce storage space for saving duplicate copies of embedding matrix. In addition, the parallel training mechanism generally reduces the training time by large. The Multi-IPW and Multi-DR models inherit aforementioned merits by incorporating a multi-task learning module.

\subsection{Multi-task Inverse Propensity Weighting CVR Estimator}
\label{section:N_IPW}
Let the marginal probability $P(o_{u,i}=1)$ denote the propensity score, $p_{u,i}$, of observing an entry in $\mathbf{R}$. In practice, the real $p_{u,i}$ can not be obtained directly. Instead, we estimate the real propensity with $\hat{p}_{u,i}$. The IPW-based estimator uses $\hat{p}_{u,i}$ to inversely weight prediction loss \cite{little2002statistical, imbens2015causal, schnabel_recommendations_2016,hirano2003efficient}, 
\begin{equation}
    \begin{aligned}
        \mathcal{E}^{\text{IPW}}
        &= \frac{1}{|\mathcal{D}|}\sum_{(u,i) \in \mathcal{D}}\frac{o_{u,i}e(r_{u,i}, \hat{r}_{u,i})}{\hat{p}_{u,i}}.
    \end{aligned}
    \label{eq:IPW}
\end{equation}

Typically, $\hat{p}_{u,i}$ is learned via an independent logistic regression model \cite{austin2011introduction}. In Fig. \ref{fig:model_overview}, we propose the Multi-IPW model which leverages the multi-task learning framework to simultaneously learn the propensity score (i.e., CTR in Multi-IPW) with CVR. Hence, the loss function of Multi-IPW estimator can be written as follows,
\begin{equation}
    \begin{aligned}
        & \mathcal{E}^{\text{Multi-IPW}}(\mathcal{X}_{\mathcal{O}};\theta_{\text{CTR}}, \theta_{\text{CVR}}, \Phi) \\
        &= \frac{1}{|\mathcal{D}|}\sum_{(u,i) \in \mathcal{D}}\frac{o_{u,i}e(r_{u,i}, \hat{r}_{u,i}(\vec{x}_{u,i};\theta_{\text{CVR}},\Phi))}{\hat{p}_{u,i}(\vec{x}_{u,i};\theta_{\text{CTR}},\Phi)},
    \end{aligned}
\end{equation}
where $\Phi$ represents the shared embedding parameters. $\theta_{\text{CVR}}$ and $\theta_{\text{CTR}}$ are neural network parameters of CVR task and CTR task, respectively. $e(r_{u,i}, \hat{r}_{u,i})$, parameterized by $\theta_{\text{CVR}}$ and $\Phi$, is the cross entropy loss of true CVR label $r_{u,i}$ and predicted CVR score $\hat{r}_{u,i}$. We use the predicted CTR score $\hat{p}_{u,i}$, parameterized by $\theta_{\text{CTR}}$ and $\Phi$, as propensities. $\mathcal{D}$ denotes all the data in the exposure space. $\mathcal{X}_\mathcal{O}$ is the input feature vectors in $\mathcal{O}$.

we formally derive the bias of Multi-IPW and prove it is unbiased given the propensities are accurately estimated.

\begin{figure*}[t]
    \centering
    \includegraphics[width=130mm]{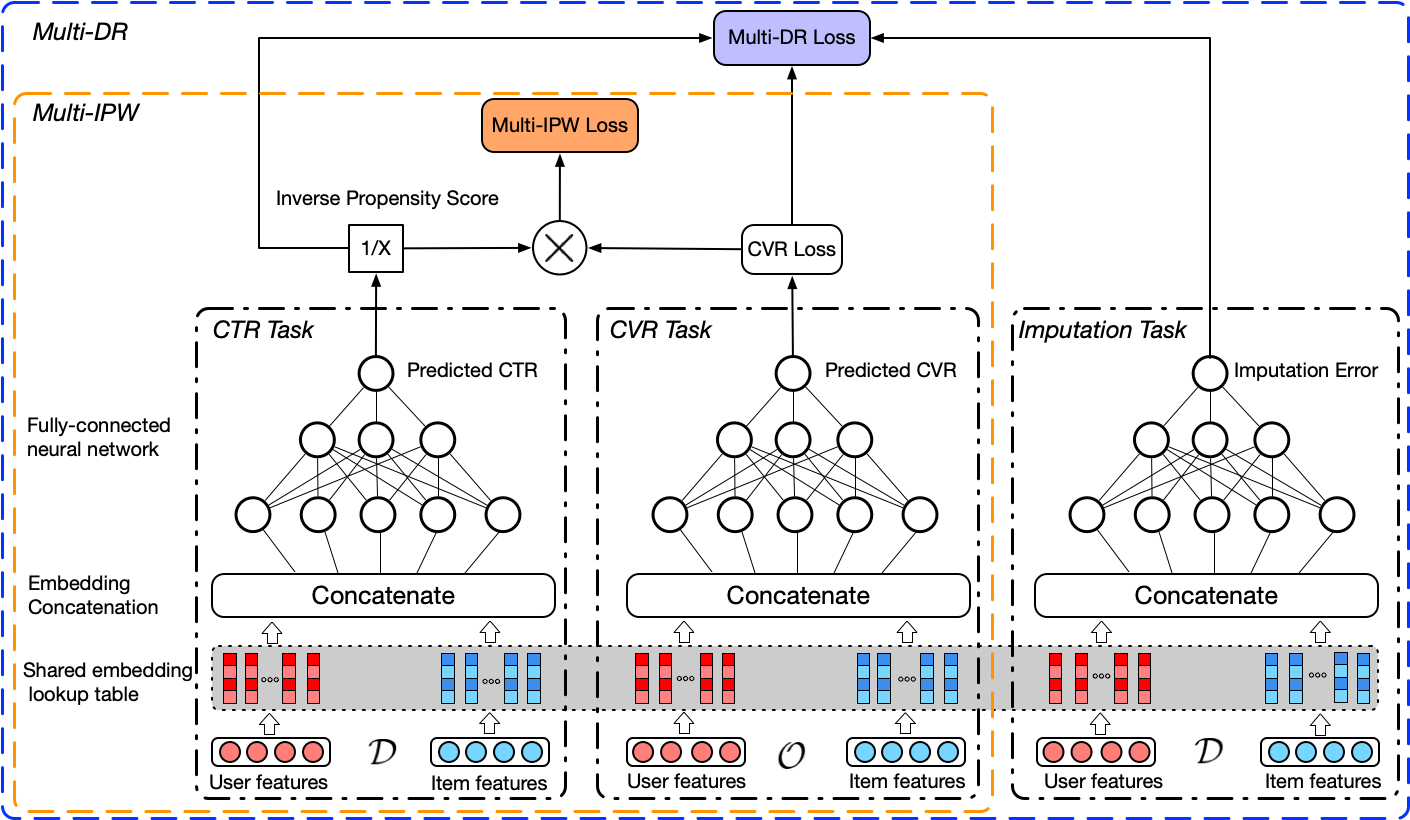}
    \caption{Multi-Inverse Propensity Weighting estimator and Multi-Doubly Robust estimator. The Multi-DR estimator augments Multi-IPW with an imputation model. We use predicted CTR as propensity scores in the Multi-IPW estimator. In the multi-task learning module, the CTR task, CVR task, and Imputation task are chained together via parameter sharing.}
    \label{fig:model_overview}
\end{figure*}

\begin{theorem}
Given the true propensities $\mathbf{P}$  and the true conversion label matrix $\mathbf{R}$, the  Multi-IPW CVR estimator gives unbiased CVR prediction when estimated propensity scores are accurate $\hat{p}_{u,i}=p_{u,i}$,
\begin{align}
    |\mathbb{E}_{\mathbf{O}}[\mathcal{E}^{\text{Multi-IPW}}] - \mathcal{P}|=0.
\end{align}
\end{theorem}
\begin{proof}
\begin{equation}
\begin{aligned}
    &|\mathbb{E}_{\mathbf{O}}[\mathcal{E}^{\text{Multi-IPW}}] - \mathcal{P}\\ &= \biggr|\frac{1}{|\mathcal{D}|}\sum_{(u,i) \in \mathcal{D}}\mathbb{E_{\mathbf{O}}}\biggr[\frac{o_{u,i}e(r_{u,i}, \hat{r}_{u,i})}{\hat{p}_{u,i}}\biggr] - \mathcal{P}\biggr|\\
    &= \biggr|\frac{1}{|\mathcal{D}|}\sum_{(u,i) \in \mathcal{D}}\frac{p_{u,i}e(r_{u,i}, \hat{r}_{u,i})}{\hat{p}_{u,i}} - \mathcal{P}\biggr|\\
    &= \biggr|\frac{1}{|\mathcal{D}|}\sum_{(u,i) \in \mathcal{D}}e(r_{u,i}, \hat{r}_{u,i}) - \mathcal{P}\biggr|= 0.
\end{aligned}    
\end{equation}
\end{proof}
Multi-IPW estimator inherits the merits of multi-task learning: 1) better CVR prediction due to parameter sharing, and 2) reduced training time and parameter storage.
 These are clear advantages over conventional IPW-based estimators.

\begin{algorithm}
\SetAlgoLined
\DontPrintSemicolon
\SetInd{0.2em}{1.3em}
\KwIn{ Observed conversion labels $\mathcal{R}^{obs}$ and user-item features $\mathcal{X}_{\mathcal{O}} = \{\vec{x}_{u,i}\}^{\mathcal{O}}, u,i \in \mathcal{O}$}
 \While{stopping criteria is not satisfied}{
  Sample a batch $\mathcal{B}$ of user-item pairs $\{\vec{x}_{u,i}\}_{B}$ from $\mathcal{O}$\;
    \SetKwBlock{Fna}{\textnormal{Co-train CTR task and CVR task }}{}
    \Fna{
        Update $\theta_{CTR}, \Phi$ by descending along the gradients $\triangledown_{\theta_{CTR}}\mathcal{E}^{Multi-IPW},\triangledown_{\Phi}\mathcal{E}^{Multi-IPW}$\;
      Update $\theta_{CVR}, \Phi$ by descending along the gradients $\triangledown_{\theta_{CVR}}\mathcal{E}^{Multi-IPW},\triangledown_{\Phi}\mathcal{E}^{Multi-IPW}$\;
    } 

}
 \caption{\textbf{Multi-Inverse Propensity  Weighting}}
\end{algorithm}

\subsection{Multi-task Doubly Robust CVR Estimator}
\label{section:dr}
The IPW-based models are unbiased contingent on accurately estimated propensities (i.e., $\hat{p}_{u,i}=p_{u,i}$). In practice, this condition is too restricted. To address this issue, doubly robust estimator is introduced by previous works \cite{dudik2011doubly,wang_doubly_2019,vermeulen2015bias,farajtabar2018more}. 

Wang \textit{et al.} \cite{wang_doubly_2019} proposed a joint learning approach for training a doubly robust estimator, and introduced two models: 1) a prediction model $\hat{r}_{u,i} = f_{\theta}(\vec{x}_{u,i})$, and 2) an imputation model $\hat{e}_{u,i} = g_{\phi}(\vec{x}_{u,i})$. The prediction model, parameterized by $\theta$, aims to predict the ratings, and its performance is evaluated by $e_{u,i}=e(r_{u,i},\hat{r}_{u,i}),  (u,i) \in \mathcal{D}$. The imputation model, parameterized by $\phi$, aims to estimate the prediction error $e_{u,i}$ with $\hat{e}_{u,i}$. Its performance is assessed by $\delta_{u,i}=e_{u,i}-\hat{e}_{u,i}$. The feature vector $\vec{x}_{u,i}$ encodes all the information about the user $u$ and the item $i$, $(u,i) \in \mathcal{D}$. Then, we can formulate the loss of doubly robust estimator as,
\begin{equation}
    \mathcal{E}^{\text{DR}} = \frac{1}{|\mathcal{D}|}\sum_{(u,i) \in \mathcal{D}}\biggr(\hat{e}_{u,i}+\frac{o_{u,i}\delta_{u,i}}{\hat{p}_{u,i}}\biggr),
\end{equation}
Similarly, we propose the Multi-DR estimator which augments Multi-IPW estimator by including an imputation model estimating the prediction error $e_{u,i}$. Multi-DR optimizes the following loss,
\begin{equation}
\begin{aligned}
    & \mathcal{E}^{\text{Multi-DR}} (\mathcal{X};\theta_{\text{CTR}}, \theta_{\text{CVR}}, \theta_{\text{Imp}}, \Phi) \\
    &= \frac{1}{|\mathcal{D}|}\sum_{(u,i) \in \mathcal{D}}\biggr(\hat{e}_{u,i}(\vec{x}_{u,i};\theta_{\text{Imp}},\Phi)+\frac{o_{u,i}\delta_{u,i}(\vec{x}_{u,i};\theta_{\text{CVR}},\theta_{\text{Imp}},\Phi)}{\hat{p}_{u,i}(\vec{x}_{u,i};\theta_{\text{CTR}},\Phi)}\biggr),
\end{aligned}    
\end{equation}
where $\Phi$ represents the shared embedding parameters among CTR task, CVR task, and imputation task. $\theta_{\text{CTR}}$, $\theta_{\text{CVR}}$, $\theta_{\text{Imp}}$ are neural network parameters of CTR task, CVR task, imputation task, respectively. $\hat{p}_{u,i}$ is the propensity (i.e., predicted CTR score) given by CTR task. Estimated prediction error $\hat{e}$, parameterized by $\theta_{\text{Imp}}$ and $\Phi$, is given by imputation task. $ \delta_{u,i} = e_{u,i} - \hat{e}_{u,i}$ is the error deviation.

We can formally derive the bias of Multi-DR and prove it is unbiased if either true propensity scores or true prediction errors are accurately estimated (i.e., $\Delta_{u,i}=0$ or $\delta_{u,i}=0$).

\begin{theorem}
Given the true propensities $\mathbf{P}$ and the true conversion label matrix $\mathbf{R}$, the  Multi-DR CVR estimator gives unbiased CVR prediction when either estimated propensity scores are accurate $\Delta_{u,i} = \frac{p_{u,i}-\hat{p}_{u,i}}{\hat{p}_{u,i}}= 0$ or the estimated prediction errors are accurate $\delta_{u,i} =e_{u,i}-\hat{e}_{u,i}=0$,
\begin{align}
    &|\mathbb{E}_{\mathbf{O}}[\mathcal{E}^{\text{Multi-DR}}] - \mathcal{P}|=0.
\end{align}
\end{theorem}

\begin{proof}
\begin{equation}
\begin{aligned}
    &|\mathbb{E}_{\mathbf{O}}[\mathcal{E}^{\text{Multi-DR}}]-\mathcal{P}|\\&=\frac{1}{|\mathcal{D}|}\biggr|\sum_{(u,i) \in \mathcal{D}}\biggr(\hat{e}_{u,i}+\mathbb{E}_{\mathbf{O}}\biggr[\frac{o_{u,i\delta_{u,i}}}{\hat{p}_{u,i}}\biggr]\biggr) - \mathcal{P}\biggr| \\
    &= \frac{1}{|\mathcal{D}|}\biggr|\sum_{(u,i) \in \mathcal{D}}\biggr(\hat{e}_{u,i}+\frac{p_{u,i}\delta_{u,i}}{\hat{p}_{u,i}}\biggr) - \mathcal{P}\biggr| \\
    &= \frac{1}{|\mathcal{D}|}\biggr|\sum_{(u,i) \in \mathcal{D}}\frac{(p_{u,i}-\hat{p}_{u,i})\delta_{u,i}}{\hat{p}_{u,i}}\biggr|\\
    &= \frac{1}{|\mathcal{D}|}\biggr|\sum_{(u,i) \in \mathcal{D}}\Delta_{u,i}\delta_{u,i}\biggr| = 0.
\end{aligned}
\end{equation}
\end{proof}

\begin{algorithm}
\SetAlgoLined
\DontPrintSemicolon
\KwIn{Observed conversion labels $\mathcal{R}^{obs}$ and user-item features $\mathcal{X}_{\mathcal{D}} = \{\vec{x}_{u,i}\}^{\mathcal{D}}, u,i \in \mathcal{D}$}
 \While{stopping criteria is not satisfied}{
  Sample a batch $\mathcal{B}$ of user-item pairs $\{\vec{x}_{u,i}\}_{B}$ from $\mathcal{D}$\;
  
    \SetKwBlock{Fnb}{\textnormal{Co-train CTR task, CVR task, and Imputation task }}{}
        \Fnb{
      Update $\theta_{CTR}, \Phi$ by descending along the gradients $\triangledown_{\theta_{CTR}}\mathcal{E}^{Multi-DR},\triangledown_{\Phi}\mathcal{E}^{Multi-DR}$\;
      Update $\theta_{CVR}, \Phi$ by descending along the gradients $\triangledown_{\theta_{CVR}}\mathcal{E}^{Multi-DR},\triangledown_{\Phi}\mathcal{E}^{Multi-DR}$\;
      Update $\theta_{Imp}, \Phi$ by descending along the gradients $\triangledown_{\theta_{Imp}}\mathcal{E}^{Multi-DR},\triangledown_{\Phi}\mathcal{E}^{Multi-DR}$\;
}
 }
 \caption{\textbf{Multi-Doubly Robust}}
\end{algorithm}

\section{Experimentation}
In this section, we evaluate the performance of the proposed models with a public dataset and a large-scale production dataset collected from Mobile Taobao, the leading e-commerce platform in China. The experiments are intended to answer the following questions:
\begin{itemize}
    \item \textbf{Q1}: Do our proposed approaches outperform other state-of-art CVR estimation methods?
    \item \textbf{Q2}: Are our proposed models more efficient in industrial setting than other baseline models?
    \item \textbf{Q3}: How is the performance of our proposed models affected by hyper-parameters?
    % \item \textbf{Q4}: Can our methods boost the performance of online our production system?
\end{itemize}

\subsection{Datasets}
\label{section:dataset}
\subsubsection*{Ali-CCP \footnote{\url{https://tianchi.aliyun.com/dataset/dataDetail?dataId=408}}\cite{ma_entire_2018}} Alibaba Click and Conversion Prediction (Ali-CCP) dataset is collected from real-world traffic logs of the recommender systems in the Taobao platform. See the statistics in Table \ref{table:data_stats}.

\subsubsection*{Production sets} This industrial production dataset is collected from the Mobile Taobao e-commerce platform.
It contains 3-week transactional data (see Table \ref{table:data_stats}). Our production dataset includes 109 features, which are primarily categorized into: 1) user features, 2) item features, and 3) combination features. We further divide this  dataset into 4 subsets: Set A,  Set B, Set C, and Set D, which contain the first two days ($5\%$), the first five days ($20\%$), the first twelve days ($50\%$), and the 3-week of data ($100\%$), respectively. We use the data of the last day in each set as testing set and the remaining data as training set.

\begin{table}[ht]
\centering
  \caption{Statistics of experimental datasets}
  \label{table:data_stats}
  \begin{tabular}{llllcl}
    \toprule
    Dataset & \# Exposure & \# Click & \# Conversion & \# user & \# item\\
    \midrule
    Ali-CCP & 84M & 3.4M & 18k & 0.4M & 4.3M\\
    Set A & 1.1B & 54.5M & 0.6M & - & 22.5M\\
    Set B & 2.7B & 0.2B & 1.9M & - & 39.1M\\
    Set C & 6.0B & 0.4B & 4.3M & - & 62.6M\\
    Set D & 11.5B & 0.6B & 8.3M & - & 81.5M\\

  \bottomrule
\end{tabular}
\end{table}

% \subsection{Ali-CCP dataset preprocessing}

% \label{section:data_preprocess}
% Features in Ali-CCP dataset can be categorized into two groups: 1) single value features, and 2) multivalent features. The first group contains features like, \textit{user age}, \textit{item id}, and etc. Features are akin to \textit{user age} feature can be translated to a low-dimension vector with ease. As for the multivalent features like \textit{User historical behaviors of shop ID and count}, we need to apply weighted-sum pooling technique. For example, \textit{User historical behaviors of shop ID and count} may contains multiple shop IDs and corresponding counts in one entry. We can translate shop IDs to embedding vectors and then apply weighted sum with counts. The output will be a single low-dimension vector that encode the information of user historical behaviors.

\subsection{Baseline models}
We compare Multi-IPW model and Multi-DR model with the following baselines. Note that some baselines are causal estimators which we modify to predict the unbiased CVR, and others models are existing non-causal estimators designed for CVR predictions.

\subsubsection{Non-causal estimators}
\begin{itemize}
    \item \textbf{Base} is a naive post-click CVR model, which is a Multi-layer Perceptron (See the CVR task in Figure \ref{fig:model_overview}). Note that this is essentially an independent MLP model which takes the feature embeddings as input and predicts the CVR. The base model is trained in the click space.
    \item \textbf{Oversampling} \cite{weiss2004mining,2019arXiv191007892Z} deals with the class imbalance issue by duplicating the minority data samples (conversion=1) in training set with an oversampling rate k. In our experiment, we set k = 5. The oversampling model is trained in the click space.
    \item \textbf{ESMM} \cite{ma_entire_2018} utilizes multi-task learning methods and reduces the CVR estimation into two auxiliary tasks, i.e., CTR task and CTCVR task. ESMM is trained in the entire exposure space, and deemed as the state-of-the-art CVR estimation model in real practice.
    \item \textbf{Naive Imputation} takes all the unclicked data as negative samples. Hence, it is trained in the entire exposure space. 
\end{itemize}

\subsubsection{Causal estimators}
\begin{itemize}
    \item \textbf{Naive IPW}\cite{schnabel_recommendations_2016} is a naive IPW estimator. Note that it is not specifically designed for CVR estimation task as CVR prediction has its intrinsic issues. For example, it cannot deal with the data sparsity issue that inherently exists in CVR task.
    \item \textbf{Joint Learning DR} \cite{wang_doubly_2019}  is devised to learn from ratings that are missing not at random. In this experiment, we tailor Joint Learning DR for the CVR estimation. Similarly, Joint learning DR handles data sparsity issue poorly. 
    \item \textbf{Heuristic DR} is designed as a baseline for Multi-DR. It assumes that the unclicked items are negative samples with probability $1-\eta$, where $\eta$ is smoothing rate and it denotes the probability of having a positive label. In the experiments, we explore $\eta$ in $\{0.0005, 0.001, 0.002, 0.005, 0.01\}$ and report the best performance.
\end{itemize}

\subsection{Metrics}
In CVR prediction task, ROC AUC is a widely used metric \cite{fawcett2006introduction}. One interpretation of AUC in the context of ranking system is that it denotes the probability of ranking a random positive sample higher than a negative sample. Meanwhile, we also adopt Group AUC (GAUC) \cite{zhou2018deep}. GAUC extends AUC by calculating the weighted average of AUC grouped by page views or users,
\begin{equation}
    \text{GAUC}=\frac{\sum_{i \in U}w_i \times \text{AUC}_i}{\sum_{i \in U}w_i},
\end{equation}
where $w_i$ is exposures. GAUC is commonly recognized as a more indicative metric in real practice \cite{zhou2018deep}. In the public dataset, models are only assessed with AUC as the dataset is missing the information for computing GAUC.

\subsection{Unbiased Evaluation}
In this work, we use CTCVR-AUC/GAUC to evaluate the unbiasedness of CVR estimators \cite{ma_entire_2018}. We need to point out that testing with an unbiased dataset or randomization is generally a golden standard for unbiasedness assessment \cite{fisher1956statistical,whitehead1991general}. However, the unbiased training/testing dataset for CVR estimation is rather unobtainable in real practice.The real-world system can randomly expose items to users and generate unbiased evaluation sets for CTR estimation. But they cannot force users to randomly click on items to generate unbiased data for CVR estimations. This limitation may be further investigated in the future work.

\subsection{Experiments setup}
\subsubsection{Ali-CCP experiment}
The experiment setup on Ali-CCP mostly follows the prior work \cite{ma_entire_2018}. We set the dimension of all embedding vectors to be 18. The architecture of all these multi-layer perceptrons (MLP) in multi-task learning module are identical as $512 \times 256 \times 128 \times 32 \times 2$.  The optimizer is Adam with a learning rate $lr=0.0002$, and batch size is set to $|batch|=1024$. 

\begin{table*}[ht]
% production dataset results
\caption{Results of comparison study on Production datasets. The best scores are bold-faced in each column. Note that this table has two sections, AUC scores and GAUC scores. The rows that contain the models proposed in this paper are highlighted in color grey. }
\begin{tabular}{lllllllll}
\hline
\multicolumn{1}{l|}{}                                  & \multicolumn{2}{l}{Set A (1.1B)}                                & \multicolumn{2}{l}{Set B (2.7B)}                               & \multicolumn{2}{l}{Set C (6.0B)}                               & \multicolumn{2}{l}{Set D (11.5B)}                              \\ \cline{2-9} 
\multicolumn{1}{l|}{\multirow{-2}{*}{Model}}           & CVR                                           & CTCVR          & CVR                                           & CTCVR          & CVR                                           & CTCVR          & CVR                                           & CTCVR          \\ \hline
\multicolumn{9}{c}{AUC score}                                                                                                                                                                                                                                                                                       \\ \hline
\multicolumn{1}{l|}{Base}                              & 78.24                                         & 73.12          & 78.67                                         & 73.86          & 79.62                                         & 74.70          & 81.66                                         & 76.28          \\
\multicolumn{1}{l|}{Oversampling\cite{weiss2004mining}}                      & 78.63                                         & 73.53          & 78.72                                         & 74.09          & 79.69                                         & 74.82          & 81.77                                         & 76.30          \\
\multicolumn{1}{l|}{ESMM\cite{ma_entire_2018}}                              & 79.29                                         & 73.86          & 79.74                                         & 74.33          & 80.11                                         & 74.97          & 82.17                                         & 76.55          \\
\multicolumn{1}{l|}{Naive Imputation}              & 78.12                                         & 73.21          & 78.44                                         & 73.50          & 79.32                                         & 73.81          & 81.56                                         & 76.39          \\
\multicolumn{1}{l|}{Naive IPW\cite{schnabel_recommendations_2016}}                         & 79.23                                         & 73.82          & 79.73                                         & 74.34          & 80.14                                         & 74.92          & 82.13                                         & 76.45          \\
\multicolumn{1}{l|}{Heuristic DR}                      & 78.45                                         & 73.45          & 78.84                                         & 73.99          & 79.52                                         & 74.18          & 81.74                                         & 76.40          \\
\multicolumn{1}{l|}{Joint Learning DR\cite{wang_doubly_2019}}                 & 79.09                                         & 73.67          & 79.53                                         & 74.51          & 80.01                                         & 74.90          & 82.09                                         & 76.61          \\ \hline
\rowcolor[HTML]{C0C0C0} 
\multicolumn{1}{l|}{\cellcolor[HTML]{C0C0C0}Multi-IPW} & 79.51                                         & 73.99          & \textbf{79.85}                                & 74.81          & 80.21                                         & 75.01          & 82.57                                         & 76.89          \\
\rowcolor[HTML]{C0C0C0} 
\multicolumn{1}{l|}{\cellcolor[HTML]{C0C0C0}Multi-DR}  & \textbf{79.72}                                & \textbf{74.45} & 79.80                                         & \textbf{74.91} & \textbf{80.50}                                & \textbf{75.39} & \textbf{82.72}                                & \textbf{77.23} \\ \hline
\multicolumn{9}{c}{GAUC score}                                                                                                                                                                                                                                                                                      \\ \hline
\multicolumn{1}{l|}{Base}                              & \multicolumn{1}{c}{-}                         & 59.69          & \multicolumn{1}{c}{-}                         & 60.16          & \multicolumn{1}{c}{-}                         & 60.58          & \multicolumn{1}{c}{-}                         & 61.27          \\
\multicolumn{1}{l|}{Oversampling\cite{weiss2004mining}}                      & \multicolumn{1}{c}{-}                         & 60.17          & \multicolumn{1}{c}{-}                         & 60.28          & \multicolumn{1}{c}{-}                         & 60.59          & \multicolumn{1}{c}{-}                         & 61.30          \\
\multicolumn{1}{l|}{ESMM\cite{ma_entire_2018}}                              & \multicolumn{1}{c}{-}                         & 60.53          & \multicolumn{1}{c}{-}                         & 60.90          & \multicolumn{1}{c}{-}                         & 61.13          & \multicolumn{1}{c}{-}                         & 61.76          \\
\multicolumn{1}{l|}{Naive Imputation}              & \multicolumn{1}{c}{-}                         & 60.14          & \multicolumn{1}{c}{-}                         & 60.39          & \multicolumn{1}{c}{-}                         & 60.56          & \multicolumn{1}{c}{-}                         & 61.39          \\
\multicolumn{1}{l|}{Naive IPW\cite{schnabel_recommendations_2016}}                         & \multicolumn{1}{c}{-}                         & 60.51          & \multicolumn{1}{c}{-}                         & 60.95          & \multicolumn{1}{c}{-}                         & 61.09          & \multicolumn{1}{c}{-}                         & 61.77          \\
\multicolumn{1}{l|}{Heuristic DR}                          & \multicolumn{1}{c}{-}                         & 60.01          & \multicolumn{1}{c}{-}                         & 60.30          & \multicolumn{1}{c}{-}                         & 60.65          & \multicolumn{1}{c}{-}                         & 61.35          \\
\multicolumn{1}{l|}{Joint Learning DR\cite{wang_doubly_2019}}                 & \multicolumn{1}{c}{-}                         & 60.43          & \multicolumn{1}{c}{-}                         & 60.83          & \multicolumn{1}{c}{-}                         & 60.97          & \multicolumn{1}{c}{-}                         & 61.67          \\ \hline
\rowcolor[HTML]{C0C0C0} 
\multicolumn{1}{l|}{\cellcolor[HTML]{C0C0C0}Multi-IPW} & \multicolumn{1}{c}{\cellcolor[HTML]{C0C0C0}-} & 60.70          & \multicolumn{1}{c}{\cellcolor[HTML]{C0C0C0}-} & \textbf{61.09} & \multicolumn{1}{c}{\cellcolor[HTML]{C0C0C0}-} & 61.25          & \multicolumn{1}{c}{\cellcolor[HTML]{C0C0C0}-} & 61.98          \\
\rowcolor[HTML]{C0C0C0} 
\multicolumn{1}{l|}{\cellcolor[HTML]{C0C0C0}Multi-DR}  & \multicolumn{1}{c}{\cellcolor[HTML]{C0C0C0}-} & \textbf{60.90} & \multicolumn{1}{c}{\cellcolor[HTML]{C0C0C0}-} & 60.99          & \multicolumn{1}{c}{\cellcolor[HTML]{C0C0C0}-} & \textbf{61.52} & \multicolumn{1}{c}{\cellcolor[HTML]{C0C0C0}-} & \textbf{62.28} \\ \hline
\end{tabular}
\label{table:product_result}
\end{table*}

\begin{table}[ht]
\centering
% public dataset results
\caption{Results of comparison study on Public dataset: Ali-CCP. Experiments are repeated 10 times and mean $\pm$ 1 std of AUC scores are reported below. The best scores are bold-faced in each column. The rows that contain the models proposed in this paper are highlighted in color grey. }
\begin{tabular}{l|ll}
\hline
Model     & CVR AUC  & CTCVR AUC \\ \hline
Base     &    66.00 $\pm$ 0.37        &      62.07 $\pm$ 0.45     \\      
Oversampling \cite{weiss2004mining}     &    67.18 $\pm$ 0.32        &      63.05 $\pm$ 0.48 \\
ESMM-NS \cite{ma_entire_2018}     &    68.25 $\pm$ 0.44        &      64.44 $\pm$ 0.62                  \\
ESMM \cite{ma_entire_2018}      &    68.56 $\pm$ 0.37        &      65.32 $\pm$ 0.49                  \\ \hline
\rowcolor[HTML]{C0C0C0} 
Multi-IPW &    69.21 $\pm$ 0.42             &       65.30 $\pm$  0.50           \\
\rowcolor[HTML]{C0C0C0} 
Multi-DR  &  \textbf{69.29} $\pm$ \textbf{0.31}                &       \textbf{65.43} $\pm$  \textbf{0.34}         \\ \hline
\end{tabular}
\label{table:public_result}
\end{table}

\subsubsection{Production set experiment}
In production set experiment, we vary the dimensions of feature embedding vectors according to each feature's real size in order to minimize the memory usage. In order to have a fair comparison study, all the models in this experiment share $|batch|=10000$, MLP architecture $1024 \times 512 \times 256 \times 128 \times 32 \times 2$, adam optimizer with learning rate $lr=0.0005$. We also added $l2$ normalization to imputation model in Multi-DR, and the coefficient is $v=0.0001$.

\section{Results and Discussion}

In this section, we evaluate the proposed models and answer the questions raised in Section 4.

\subsection{Model Assessments (Q1)}

In this section, we report the experiment results in Table \ref{table:product_result}, \ref{table:public_result}. Multi-IPW and Multi-DR are clear winners over other baselines across all experiments. Meanwhile, we have the following observations:

\begin{itemize}
    \item In production dataset, Multi-IPW and Multi-DR consistently outperform Joint Learning DR \cite{wang_doubly_2019}. We reason that the performance improvement benefits from multi-task learning module, which remedies the data sparsity issue. 
    \item We notice that Multi-DR mostly has better performance than Multi-IPW.
    Recall that Multi-DR augments Multi-IPW by introducing an imputation model. Provided that $0 \leq \hat{e} \leq 2e$, the tail bound of Multi-DR is proven to be lower than that of Multi-IPW for any learned propensity score $\hat{p}_{u,i}, (u,i) \in \mathcal{O}$ \cite{wang_doubly_2019}. Therefore, Multi-DR is expected to perform better than Multi-IPW when the imputation model is well-trained.
    \item We observe that Multi-IPW/Naive IPW estimator are superior to Base in all experiments. 
    Compared with the Base, both IPW-based models introduce estimated propensities to correct the selection bias. 
    Recall that Theorem 1. ensures the CVR estimators are unbiased if the propensities are accurately estimated. 
    While in practice the estimated propensities may deviate from the real values, this control experiment  attests to “enough accuracy” of the propensity model.

\end{itemize}

The experiment results demonstrate that Multi-IPW and Multi-DR counter the selection bias and data sparsity issues in CVR estimation in a principled and highly effective way. In the next subsection, we will discuss other strengths of the proposed methods.

\subsection{Computational efficiency (Q2)}
\label{section:time_complexity}

\begin{figure*}[ht]
    \centering
    \includegraphics[width=150mm]{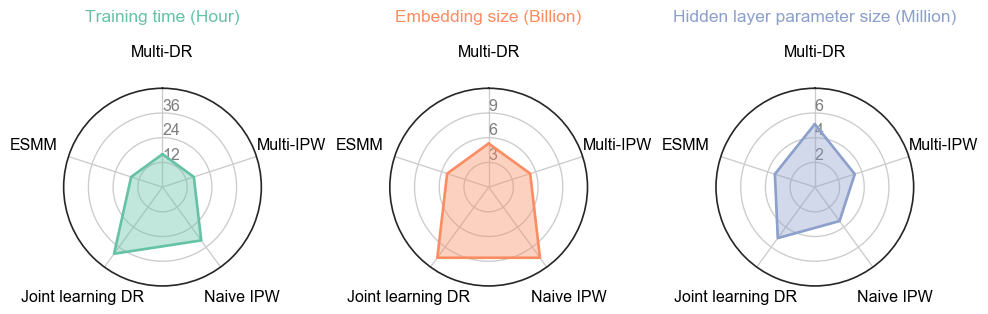}
    \caption{Computational cost of Multi-IPW and Multi-DR. The left subplot reveals the hours needed to complete one epoch of training. The middle subplot shows the size of embedding parameters of each model. The right subplot shows the size of hidden layer parameters of each model. Note that the proposed models achieve the best prediction performance, while have the lowest computational cost.}
    \label{fig:radar_chart}
\end{figure*}

In this section, we study the computational efficiency of the proposed models against the baselines under industrial setting. We summarize the records of training time and parameter space size of each model in Fig. \ref{fig:radar_chart}, and the cluster configuration in Table \ref{table:config}.

We observe that Multi-IPW and Multi-DR require less or equivalent training time compared with other baselines. Recall that multi-task learning method co-trains multiple tasks simultaneously as if they were one task. We can expect the training time being greatly shortened. Meanwhile, our methods are also economical in memory usage due to parameter sharing in the multi-task learning module. 

\begin{table}[ht]
\centering
\caption{Distributed cluster configuration}
\begin{tabular}{l|cc}
\hline
Cluster configuration & \multicolumn{1}{l}{Parameter Server} & \multicolumn{1}{l}{Worker} \\ \hline
\# instances            & 4                                    & 100                             \\
\# CPU           & 28 cores                             & 440 cores                       \\
\# GPU\tablefootnote{GPU specs: Tesla P100-PCIE-16G}               & -                                    & 25 cards                              \\
MEMORY (GB)           & 40                                   & 1000                            \\ \hline
\end{tabular}
\label{table:config}
\end{table}

\subsection{Hyper-parameters in model implementation}

 IPW bound $\tau$ is a hyper-parameter introduced in our model implementation to handle high variance of propensities. IPW bound clamps the propensities if the values are greater than the predefined threshold. A plausible IPW bound value is typically confined by $\tau \in [propensity_{min},propensity_{mean}]$. IPW bound percentage can be calculated as $\tau\%=\frac{\tau-propensity_{min}}{propensity_{mean}-propensity_{min}}$.
 In Multi-DR, imputation model will introduce the unclicked items to the training set. Empirically, most of the unclicked items will not be purchased by customers even if they were clicked. Therefore, including these unclicked items in training set will skew the data distribution and make the class imbalance issue worse. Therefore, Instead of adding all the unclicked samples, we under-sample them with a sampling rate $\lambda$. For example, if the number of clicked samples ($N_{clicked}$) is 100 and the batch size is 1000, $\lambda=1.5$ means that after under-sampling the samples we used to train Multi-DR is $N_{clicked} \times \lambda = 100 \times 1.5 = 150$. Note that without under-sampling, Multi-DR takes all samples in the batch as training samples.

\subsection{Empirical study on hyper-Parameter sensitivity (Q3)}

\begin{figure*}[ht]
    \centering
    \includegraphics[width=150mm]{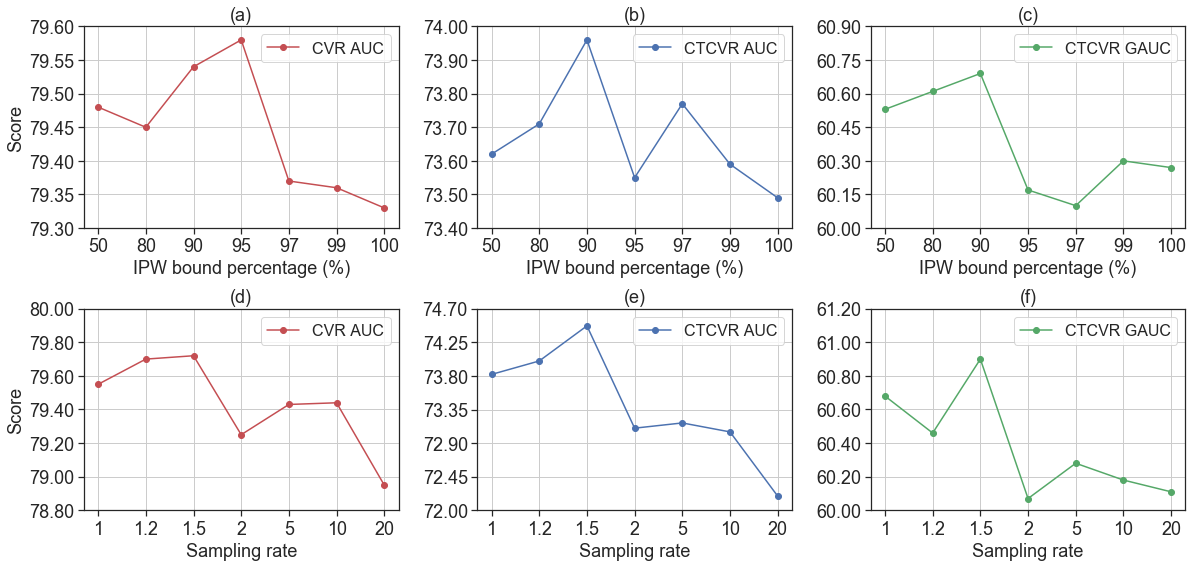}
    \caption{Results of parameter sensitivity experiments. }
    \label{fig:parameter_sensitivity}
\end{figure*}

In this section, we investigate how Multi-IPW and Multi-DR are affected by two important hyper-parameters in our model implementation, IPW bound $\tau$ and sampling rate $\lambda$. We evaluate the performance of Multi-IPW with varying IPW bound $\tau$. We observe that in Fig. \ref{fig:parameter_sensitivity}, when IPW bound $50\%<\tau\%<90\%$, prediction performance eventually improves as IPW bound increases. We can clearly see the performance drop of CVR AUC if the threshold is greater than $95\%$. We reason that larger IPW bound allows undesired higher variability of propensity scores, which may lead to sub-optimal prediction performance.

We evaluate the performance of Multi-DR with varying sampling rate $\lambda$. We observe that $\lambda=1.5$ produces the best prediction, and  the model performance starts decreasing when $\lambda>5$. We argue that, as the sampling rate increases, more unclicked samples are included to our training set, and it inevitably worsens the class imbalance issue, which typically causes predictive models to generalize poorly. On the contrary, introducing a small number of unclicked samples from the imputation model can boost our CVR prediction (see figure (d) when $\lambda \in [1.0, 1.5]$).

\section{Conclusion and future works}
In this paper, we proposed Multi-IPW and Multi-DR CVR estimators for industrial recommender system. Both CVR estimators aim to counter the inherent issues in practice: 1) selection bias, and 2) data sparsity. Extensive experiments with billions of data samples demonstrate that our methods outperform the state-of-the-art CVR predictive models, and handle CVR estimation task in a principled, highly effective and efficient way. Although our methods are devised for CVR estimation, the idea can be generalized to debiasing CTR estimation by exploiting the sequential pattern "item pool $\rightarrow$ exposure $\rightarrow$ click".

% \begin{acks}
% \end{acks}


% \appendix

\bibliographystyle{unsrt}
\bibliography{references}

\end{document}